\documentclass[12pt]{article}
\usepackage{latexsym}
\usepackage{epsfig}

\usepackage{amsmath}
\usepackage{amsfonts}
\usepackage{amsthm}
\usepackage{amssymb}
\usepackage{url}
\usepackage{hyperref}

\newcommand{\ket}[1]{\left | \, #1 \right \rangle}

\newtheorem{lemma}{Lemma}
\newtheorem{algorithm}{Algorithm}
\newtheorem{corollary}{Corollary}

\title{Interactive proofs with efficient quantum prover for recursive Fourier sampling}
\author{Matthew McKague \\ Centre for Quantum Technologies \\ National University of Singapore \\Ê\url{matthew.mckague@nus.edu.sg}}

\begin{document}

\maketitle

\begin{abstract}
We consider the recursive Fourier sampling problem (RFS), and show that there exists an interactive proof for RFS with an efficient classical verifier and efficient quantum prover.  

\end{abstract}

\section{Introduction}
The Recursive Fourier Sampling (RFS) problem is an oracle based decision problem that was proposed by Bernstein and Vazirani in \cite{Bernstein:1997:Quantum-Complex}.  Historically, RFS was the first problem that showed a relativized separation between $\mathsf{BQP}$ and $\mathsf{P}$.

Along with their definition of RFS, Bernstein and Vazirani also proved a lower bound for the classical query complexity and an upper bound for the quantum query complexity, establishing the $\mathsf{P}$ versus relativized $\mathsf{BQP}$ separation.  Later Aaronson \cite{Aaronson:2003:Quantum-Lower-B} proved matching lower bounds for the quantum query complexity of RFS.  Johnson \cite{Johnson:2008:Upper-and-Lower} extended and improved these results, including bounds on the polynomial degree of RFS.  Hallgren and Harrow \cite{Hallgren:2008:Superpolynomial} subsequently generalized the recursive structure of the problem to allow a super-polynomial speedup, relative to an oracle, based on randomly generated quantum circuits.

Bernstein and Vazirani also showed that RFS is not in $\mathsf{NP}$.  Aaronson  \cite{Aaronson:2003:Quantum-Lower-B} claims that the proof can be extended to show that RFS is not in $\mathsf{MA}$, and it has been conjectured by Vazirani and others that it is not even in $\mathsf{PH}$.  This gives an oracle relative to which $\mathsf{BQP}$ is not in $\mathsf{MA}$, and hope that a relativized separation may be found between $\mathsf{BQP}$ and $\mathsf{PH}$.

In the remainder of this section we recall the definition of the recursive Fourier sampling problem and the optimal classical and quantum algorithms.  Our main contribution is in section 2, where we give an algorithm for an interactive proof for RFS and show how the standard quantum algorithm for RFS can be adapted to be used as an efficient prover for the interactive proof.

\subsection{The recursive Fourier sampling problem}
We begin by defining a type of tree.  Let $n,l$ be a positive integers and consider a symmetric tree where each node, except the leaves, has $2^{n}$ children, and the depth is $l$.  Let the root be labeled by $(\emptyset)$.  The root's children are labelled $(x_{1})$ with $x_{1} \in \{0,1\}^{n}$.  Each child of $(x_{1})$ is, in turn, labelled $(x_{1}, x_{2})$ with $x^{2} \in \{0,1\}^{n}$.  We continue until we have reached the leaves, which are labelled by $(x_{1}, \dots, x_{l})$.  Thus each node's label can be thought of as a path describing how to find the node from the root.

Now we add the Fourier component to tree.  We begin by fixing an efficiently computable function $g: \{0,1\}^{n} \rightarrow \{0,1\}$.\footnote{
For concreteness $g(s)$ may be taken to be the sum of the components of $s$, modulo 3, i.e. the number of 1s, or Hamming weight, of $s$ modulo 3.}
 With each node of the tree $(x_{1}, \dots, x_{k})$ we associate a ``secret'' string $s_{(x_{1}, \dots, x_{k})} \in \{0,1\}^{n}$.  These secrets are promised to obey 
\begin{equation}
g\left(s_{(x_{1}, \dots, x_{k})}\right)  = s_{(x_{1}, \dots, x_{k-1})} \cdot x_{k}
\end{equation}
for $k \geq 1$, and the inner product taken modulo 2.  (Here we take $s_{(x_{1}, \dots, x_{k-1})}$ to mean $s_{(\emptyset)}$ if $k=1$.)  In this way, each node's secret encodes one bit of information about its parent's secret.

Now suppose that we are given an oracle $A: (\{0,1\}^{n})^{l} \rightarrow \{0,1\}$ which behaves as\footnote{
	In fact the values $s_{(x_{1}, \dots, x_{l})}$ are not necessary, since we can only ever learn $g(s_{(x_{1}, \dots, x_{l})})$.  However, including them in the definition eliminates special cases at the level $l$.
}
\begin{equation}
A\left(x_{1}, \dots, x_{l}\right) = g\left(s_{(x_{1}, \dots, x_{l})}\right).
\end{equation}
Note that $A$ works for the leaves of the tree \emph{only}.  Our goal is to find $g\left(s_{(\emptyset)}\right)$.  This is the recursive Fourier sampling problem (RFS).

At this point we wish to discuss the recursive nature of the RFS problem.  First, note that the subtree rooted at any node obeys the same promises as the whole tree.  Thus each subtree defines an RFS problem.  (The subtree rooted at a leaf is a trivial problem, where the oracle just returns the solution.)  Thus we have a methodology for solving RFS problems:  solve subtrees in order to determine information about the root's secret, then calculate the secret.  Solving subtrees uses the same method, except when we reach a leaf, where the oracle is used instead.  This is a type of top down recursive structure, where the tree is built from a number of subtrees with the same structure.

Another way of viewing the tree is from the bottom up.  Note that if we remove all the leaf nodes, truncating the tree, then the remaining tree still obeys all the required promises.  The oracle no longer returns relevant information, however.  This problem can be solved by building a new oracle using the old one.  The new oracle, given $(x_{1}, \dots, x_{l-1})$, calculates the secret $s_{(x_{1}, \dots, x_{l-1})}$ by accessing the old oracle.  The new oracle then returns $g(s_{(x_{1}, \dots, x_{l-1})})$.  An oracle constructed in this way behaves exactly as the old oracle does, but for the truncated tree.  The tree can be truncated again to depth $l-2$ and another new oracle is built on the previous one.  The result is a recursive algorithm that eventually solves the RFS problem.

With a little thought, it is easy to see that both pictures of the recursive structure of the RFS problem are essentially equivalent.  Indeed, solving a subtree rooted at $(x_{1}, \dots, x_{k})$ just means returning $g(s_{(x_{1}, \dots, x_{k})})$, which is what the oracle for the tree truncated at level $k$ does.  However, each picture can be useful in understanding the structure of algorithms.

\subsection{Classical solution}
Now let us consider the following solution to the recursive Fourier sampling problem in the classical setting.  To calculate $g\left(s_{(\emptyset)}\right)$ we must first find $s_{(\emptyset)}$.  To do so, let us define $1_{j}$ to be the $n$-bit string with a 1 in the $j$th position and 0 elsewhere.  Define $(s_{(\emptyset)})_{j}$ to be the $j$th bit of $s_{(\emptyset)}$, which is given by $s_{(\emptyset)} \cdot 1_{j}$.  These values are given by the solution to the RFS problem defined on the subtree rooted at $(x_{1})$ with $x_{1} = 1_{j}$.  After determining these values for $j = 1Ê\dots n$ we have determined $s_{(\emptyset)}$.  Hence we obtain the following algorithm:

\begin{algorithm}[RFS]\label{algorithm:rfs}
Input: oracle $A$, subroutine $g$, $l,k$, $(x_{1}, \dots, x_{k})$ 
\begin{enumerate}
	\item If $k = l$ then return $A(x_{1}, \dots, x_{l})$
	\item For $j = 1 \dots n$
	\begin{enumerate}
		\item Set $(s)_{j} = \text{RFS}(A, g, l, k+1, (x_{1}, \dots, x_{k}, 1_{j}))$
	\end{enumerate}
	\item Return $g(s)$
\end{enumerate}
\end{algorithm}

%
\begin{lemma}
$\text{RFS}(A,g,l,0,())$ returns $g(s_{(\emptyset)})$.
\end{lemma}

\begin{proof}
We proceed by induction.  We claim that for $0 \leq k \leq l$
\begin{equation}
\text{RFS}(A,g,l,k,(x_{1}, \dots, x_{k})) = g(s_{(x_{1},\dots, x_{k})}).
\end{equation}
For $k=l$ this is true by the definition of the oracle.

Now suppose that $0 \leq k < l$.  By induction, step (a) sets $(s)_{j} = g\left(s_{(x_{1}, \dots, x_{k}, 1_{j})}\right)$ which is promised to equal $s_{(x_{1}, \dots, x_{k})} \cdot 1_{j}$.  Thus $s = s_{(x_{1}, \dots, x_{k})}$ and the function returns $g(s_{(x_{1},\dots, x_{k})})$, as claimed.  Calling $RFS(A,g,l,0,())$ then returns $g\left(s_{(\emptyset)}\right)$.

\end{proof}

The query complexity of this algorithm is $n^{l}$, which can be seen since each call to $RFS$ calls itself $n$ times, and the depth of recursion is $l$.  

This solution is due to Bernstein and Vazirani \cite{Bernstein:1997:Quantum-Complex}, who also gave a matching lower bound on the number of queries, so the query complexity is $\Theta(n^{l})$.


\subsection{Quantum solution}\label{sec:quantumsolution}
We now consider the analogous quantum problem, where the oracle $A$ allows quantum access according to
\begin{equation}
	A\ket{x_{1}}\dots \ket{x_{l}}\ket{y} = \ket{x_{1}} \dots \ket{x_{l}} \ket{y \oplus g(s_{(x_{1}, \dots, x_{l})})}.
\end{equation}
In addition, we have an efficient quantum algorithm which calculates $g$ as
\begin{equation}
G\ket{s}\ket{y} = \ket{s}\ket{y \oplus g(s)}.
\end{equation}
The main idea behind the algorithm is to use the fact that $H^{\otimes n}$ transforms $\ket{y}$ into $\sum_{x} (-1)^{x \cdot y} \ket{x}$ and vice versa.  We use phase feedback and a call to $A$ to create the state $\sum_{x_{l}} (-1)^{s_{(x_{1}, \dots, x_{l-1})} \cdot x_{l}} \ket{x}$, then apply $H^{\otimes n}$ to obtain $\ket{s_{(x_{1}, \dots, x_{l-1})}}$.  After calculating $g(s_{(x_{1}, \dots, x_{l-1})})$, we uncompute $\ket{s_{(x_{1}, \dots, x_{l-1})}}$ to disentangle this register from other registers.

\begin{algorithm}[QRFS]\label{algorithm:qrfs}
Input: oracle $A$, subroutine $G_{j}$, $l$, $k$, quantum registers $\mathcal{X}_{1}, \dots, \mathcal{X}_{k}$, $\mathcal{Y}$.
\begin{enumerate}
	\item If $k= l$ then apply $A$ to $(\mathcal{X}_{1} \dots \mathcal{X}_{l})$ and $\mathcal{Y}$, then return.
	\item Introduce ancilla $\mathcal{X}_{k+1}$ in the state $\frac{1}{\sqrt{2^{n}}} \sum_{x_{k+1} \in \{0,1\}^{n}} \ket{x_{k+1}}_{\mathcal{X}_{k+1}}$.
	\item Introduce ancilla $\mathcal{Y}^{\prime}$ in state $\frac{1}{\sqrt{2}}\left(\ket{0}_{\mathcal{Y}^{\prime}} - \ket{1}_{\mathcal{Y}^{\prime}}\right)$.	
	\item Call  $QRFS(A,G,l,k+1,(\mathcal{X}_{1} \dots \mathcal{X}_{k+1}), \mathcal{Y}^{\prime})$
	\item Apply $H^{\otimes n}$ on register $\mathcal{X}_{k+1}$.
	\item Apply $G$ to $\mathcal{X}_{k+1}$ and $\mathcal{Y}$.
	\item Apply $H^{\otimes n}$ on register $\mathcal{X}_{k+1}$.
	\item Call $QRFS(A,G,l,k+1,(\mathcal{X}_{1} \dots \mathcal{X}_{k+1}), \mathcal{Y}^{\prime})$
	\item Discard $\mathcal{X}_{k+1}$ and $\mathcal{Y}^{\prime}$.
	\item Return
\end{enumerate}
\end{algorithm}

%
\begin{lemma}
$\text{QRFS}(A,G,l,0,(), \ket{0})$ returns $\ket{g(s_{(\emptyset)})}$.
\end{lemma}

\begin{proof}
To analyze the correctness of the algorithm we proceed by induction with the hypothesis that for each $0 \leq k \leq l$,  $QRFS(A,G,l,k, (\mathcal{X}_{1}\dots \mathcal{X}_{k}), \mathcal{Y})$ applied to $\ket{x_{1}}\dots\ket{x_{k}}\ket{y}$ gives
\begin{equation}
	\ket{x_{1}}\dots\ket{x_{k}}\ket{y \oplus g\left(s_{(x_{1}, \dots, x_{k})}\right)}.
\end{equation}
(The output for general input states is determined by linearity.)  This is true for $k=l$ by definition of the oracle.

Now let $0 \leq k < l$.  We introduce the state $\frac{1}{\sqrt{2}}\left(\ket{0} - \ket{1}\right)$ in step 3 in order to use phase kickback.  By hypothesis, then, step 4 introduces phase of $(-1)^{g(s_{(x_{1}, \dots, x_{k+1})})}$, which is the same as 
$(-1)^{s_{(x_{1}, \dots, x_{k})} \cdot x_{k+1}}$.  The state after step 4 is thus
\begin{equation}
\frac{1}{\sqrt{2^{n}}}\sum_{x_{k+1}} (-1)^{s_{(x_{1}, \dots, x_{k})} \cdot x_{k+1}}\ket{x_{1}}\dots \ket{x_{k+1}} \ket{y} \ket{-}.
\end{equation}
After step 5 the state then becomes
\begin{equation}
	\ket{x_{1}}\dots \ket{x_{k}} \ket{s_{(x_{1}, \dots, x_{k})}}\ket{y} \ket{-}.
\end{equation}
Step 6 changes $\ket{y}$ to $\ket{y \oplus g\left(s_{(x_{1}, \dots, x_{k})}\right)}$.  Next, steps 7 and 8 uncompute the $\mathcal{X}_{k+1}$ register, so it and the $\mathcal{Y}^{\prime}$ register are returned to their original state when they are discarded in step 9.  Thus we obtain the required output state.

From the inductive hypothesis, we see that $QRFS(A,G,l,0, (), \mathcal{Y})$ applied to $\ket{0}$ gives $\ket{g\left(s_{(\emptyset)}\right)}$.  
\end{proof}

Note that if we apply $QRFS(A,G,l,k, (\mathcal{X}_{1}\dots \mathcal{X}_{k}), \mathcal{Y})$ to $\ket{x_{1}}\dots\ket{x_{k}}\ket{y}$, and stop after step 5, we obtain $\ket{s_{(x_{1}, \dots, x_{k}})}$ in the $\mathcal{X}_{k+1}$ register.  This is equivalent to solving the RFS problem defined by the subtree rooted at $(x_{1}, \dots, x_{k})$.  Thus we can also efficiently calculate any $s_{(x_{1}, \dots, x_{k}})$.

This algorithm is due to Bernstein and Vazirani \cite{Bernstein:1997:Quantum-Complex}.  The quantum query complexity of this algorithm is $O(2^{l})$, since two recursive calls are made, and the depth is $l$.  Aaronson \cite{Aaronson:2003:Quantum-Lower-B} gave a matching lower bound.  

\subsection{Complexity implications}

In the previous sections we have kept both $n$, the size of the bit strings, and $l$, the depth of the tree.  Typically, RFS is considered with $l = \log_{2} n$.  In this case we obtain query complexities of $\Theta(n^{\log_{2} n})$ classically, and $O(n)$ quantumly.  Hence we obtain the relativized separation of $\mathsf{BQP} \nsubseteq \mathsf{P}$.  It is also for this value of $l$ that $\text{RFS} \notin \mathsf{MA}$.

\section{Interactive proof}
Suppose now that we have, in addition to the oracle $A$ and classical computing resources, access to a prover $P$ who has more powerful computing resources.  We have seen that for the choice $l = \log_{2} n$ there is no way of efficiently computing the solution to the recursive Fourier sampling problem directly.  In fact, since RFS is not in $\mathsf{NP}$, we cannot even efficiently verify a solution if it is given (along with a short proof.)  We will show, however, that by \emph{interacting} with $P$ we can efficiently compute the solution, or detect if $P$ is giving false information.  That is, RFS is in $\mathsf{IP}$.

It is not surprising that RFS is in $\mathsf{IP}$, indeed $\mathsf{BQP} \subseteq \mathsf{IP}$ in the unrelativized world.  However, $\mathsf{IP}$ is defined with an computationally unbounded prover, and here we will see that for RFS it suffices to have an efficient quantum prover.

\subsection{Classical verifier}

Looking back at the original classical solution to RFS in Algorithm~\ref{algorithm:rfs} we see that there are two steps:  find $s_{(x_{1}, \dots, x_{k})}$ and calculate $g(s_{(x_{1}, \dots, x_{k})})$.  The difficult part is finding $s_{(x_{1}, \dots, x_{k})}$, so we can instead ask $P$ to do this for us.  Since we do not trust $P$, we should perform some type of test to see whether $P$ has really given us the correct value of $s_{(x_{1}, \dots, x_{k})}$.

Suppose that $P$ gives us a string $s^{\prime}_{(x_{1}, \dots, x_{k})}$.  In principle we can detect the case $s^{\prime}_{(x_{1}, \dots, x_{k})} \neq s_{(x_{1}, \dots, x_{k})}$ by instead looking at whether $s^{\prime}_{(x_{1}, \dots, x_{k})} \cdot x_{k+1} = s_{(x_{1}, \dots, x_{k})} \cdot x_{k+1}$ for a randomly chosen $x_{k+1} \in \{0,1\}^{n}$.  If $s^{\prime}_{(x_{1}, \dots, x_{k})} \neq s_{(x_{1}, \dots, x_{k})}$ then $s^{\prime}_{(x_{1}, \dots, x_{k})} \cdot x_{k+1} \neq s_{(x_{1}, \dots, x_{k})} \cdot x_{k+1}$ for half of all strings $x_{k+1} \in \{0,1\}^{n}$.  If we check for $c$ independently chosen strings $x_{k+1}$ then we will fail to detect $P$'s deception with probability $2^{-c}$.

Since we have $s^{\prime}_{(x_{1}, \dots, x_{k})}$, we can calculate $s^{\prime}_{(x_{1}, \dots, x_{k})} \cdot x_{k+1}$ readily, but how do we find $s_{(x_{1}, \dots, x_{k})} \cdot x_{k+1}$?  We use recursion:  we ask $P$ for $s_{(x_{1}, \dots, x_{k+1})}$ and again perform the test.  After the recursion is deep enough, we can query the oracle directly to find $s_{(x_{1}, \dots, x_{l-1})} \cdot x_{l} = A(x_{1}, \dots, x_{l})$.

\begin{algorithm}[VERIFIER]
Input: Oracle $A$, Prover $P$,  subroutine $g$, total number of levels $l$, current level $k$, queries $(x_{1}, \dots x_{k})$.
\begin{enumerate}
	\item If $k = l$ then return $A(x_{1}, \dots, x_{l})$

	\item Query $P$ for $s^{\prime}_{(x_{1}, \dots, x_{k})}$
	\item Repeat 3 times:
	\begin{enumerate}
		\item Randomly choose $x_{k+1} \in \{0,1\}^{n}$
		\item Set $a= \text{VERIFIER}(A, P,g,l,k+1,(x_{1}, \dots, x_{k+1}))$
		\item If $a \neq s^{\prime}_{(x_{1}, \dots, x_{k})} \cdot x_{k+1}$ then abort
	\end{enumerate}
	\item Return $g(s^{\prime}_{x_{1}, \dots, x_{k}})$
\end{enumerate}
\end{algorithm}

We must also specify what the behaviour of $P$ should be.  When $P$ receives a query $(x_{1}, \dots, x_{k})$, it should return a string $s^{\prime}_{(x_{1}, \dots, x_{k})} \in \{0,1\}^{n}$.  For an honest $P$ this should be equal to $s_{(x_{1}, \dots, x_{k})}$.

Clearly there exists a prover $P$ which always returns $s_{(x_{1}, \dots, x_{k})}$ when queried:  $P$ can solve a subtree of the full RFS problem, rooted at $(x_{1}, \dots, x_{k})$ to find $s_{(x_{1}, \dots, x_{k})}$.  For such a prover, we may perform $\text{VERIFIER}(A, P, g,l,0, ())$ to obtain $g\left(s_{(\emptyset)}\right)$, which is the answer to the RFS problem.  

\begin{lemma}
Let $P$ be a prover that always returns $s^{\prime}_{(x_{1}, \dots, x_{k})} = s_{(x_{1}, \dots, x_{k})}$ when queried.  Then
\begin{equation}
\text{VERIFIER}(A,P, g,l,0,()) = g(s_{(\emptyset)}).
\end{equation}
\end{lemma}

\begin{proof}
We proceed via induction, as in previous proofs.  We claim that for $0 \leq k \leq l$
\begin{equation}
\text{VERIFIER}(A,P,g,l,k,(x_{1}, \dots, x_{k})) = g(s_{(x_{1},\dots, x_{k})}).
\end{equation}
This is true for $k=l$ from the definition of the problem and line 1.  For $k < l$ we see from line 4 that the claim holds for this choice of $P$ so long as the algorithm does not abort in line (c).  

Now suppose that $0 \leq k < l$.  By induction, step (b) sets $a = g(s_{(x_{1},\dots, x_{k+1})}) = s_{(x_{1}, \dots, x_{k-1})} \cdot x_{k}$ by the definition of the problem.  Again, by the choice of $P$ this is always equal to $s^{\prime}_{(x_{1}, \dots, x_{k-1})} \cdot x_{k}$ and the algorithm never aborts.

\end{proof}

Our next concern is what the algorithm does when interacting with a prover $P$ that is not return correct strings.  The next lemma says that the algorithm returns an incorrect result with low probability; the rest of the time the algorithm aborts.  Hence we are protected from a malicious $P$.

\begin{lemma}
For any $P$, the probability that $\text{VERIFIER}(A,P,g,l,0,())$ does not abort and returns a value not equal to $g(s_{(\emptyset)})$ is at most $\frac{1}{4}$.

\end{lemma}

\begin{proof}
Let $p_{k}$, $0 \leq k \leq l$ be the probability that VERIFIER does not abort and returns a value that is \emph{not} equal to $g(s_{(x_{1}, \dots, k_{k})})$.  To be precise, we should specify what $P$ does, since $p_{k}$ can depend on $P$'s behaviour.  Let us then take the maximal $p_{k}$ over all choices of $P$, which will give us an upper bound for any particular $P$.

We proceed by induction with the hypothesis that $p_{k} \leq 1/4$.  For $k = l$ this true since $A(x_{1}, \dots x_{l}) = g(s_{(x_{1}, \dots, x_{l})})$ so $p_{l} = 0$.

Now suppose that $0 \leq k < l$.  There are two cases.  First, $P$ returns $s^{\prime}_{(x_{1}, \dots, x_{k})}$ such that $g(s^{\prime}_{(x_{1}, \dots, k_{k})}) =  g(s_{(x_{1}, \dots, k_{k})})$.  In this case VERIFIER always returns the correct value, so $p_{k} = 0$.  Note that VERIFIER may still abort, since $s^{\prime}_{(x_{1}, \dots, x_{k})}$ may not equal $s_{(x_{1}, \dots, x_{k})}$, or the recursion may return an incorrect result.

Now consider the case where $P$ returns some $s^{\prime}_{(x_{1}, \dots x_{k})}$ such that $g(s^{\prime}_{(x_{1}, \dots x_{k})}) \neq g(s_{(x_{1}, \dots x_{k})})$.  If the algorithm does not abort in line (c) then one of two things happened:  either $s^{\prime}_{(x_{1}, \dots x_{k})} \cdot x_{k+1} = s_{(x_{1}, \dots x_{k})} \cdot x_{k+1}$ for this choice of $x_{k+1}$, which happens with probability $\frac{1}{2}$, or $s^{\prime}_{(x_{1}, \dots x_{k})} \cdot x_{k+1} \neq s_{(x_{1}, \dots x_{k})} \cdot x_{k+1}$ and at the same time the recursion in (b) returns an incorrect value.  The latter two events occur together with probability at most $\frac{p_{k+1}}{2}$.  Thus, for each of the three repetitions, the probability of not aborting is at most $\frac{1}{2}( 1+ p_{k+1})$ and the overall chance of not aborting is
\begin{equation}
p_{k} \leq \frac{1}{2^{3}}\left(1 + p_{k+1}\right)^{3}.
\end{equation}
By the induction hypothesis, $p_{k+1} \leq \frac{1}{4}$.  Then 
\begin{equation}
p_{k} \leq \frac{1}{8} \left( 1 + \frac{1}{4} \right)^{3} = \frac{125}{512} \leq \frac{1}{4}.
\end{equation}
Using induction, we obtain the result for $0 \leq k \leq l$.  In particular, $p_{0} \leq 1/4$, which is the desired result.

\end{proof}

\begin{corollary}
$\text{VERIFIER}(A,P,g,l,0,())$ solves the recursive Fourier sampling problem with completeness 1 and soundness $1/4$.
\end{corollary}

The algorithm uses $3^{l}$ queries to the oracle and fewer than $3^{l}$ to the prover.  With the choice $l = \log_{2}n$ this is polynomial in $n$.

\subsection{Quantum prover}

Although the class $\mathsf{IP}$ is usually defined for a computationally unbounded prover, it happens that for RFS and the verifier presented above, the prover can be an efficient quantum prover.  That is, if the prover is quantum then it need only make a polynomial number of queries to the oracle.

As mentioned in section~\ref{sec:quantumsolution}, algorithm~\ref{algorithm:qrfs} makes at most $2^{l}$ queries to find $s_{(\emptyset)}$ and can be adapted to find $s_{(x_{1}, \dots, x_{k})}$ using at most this number of queries.  We can thus readily adapt algorithm~\ref{algorithm:qrfs} to create a prover for the above classical verifier.  The total number queries that the prover will make is less than $3^{l}2^{l} \leq n^{2.58}$.

\section{Discussion}

Recently there has been interest in interactive proofs in a quantum context and, most relevant here, interaction between quantum provers and classical verifiers.  Work in this direction began with Mayers and Yao \cite{Mayers:2004:Self-testing-qu}, and Magniez et al.\ \cite{Magniez:2006:Self-testing-of} who showed how to classically verify (with certain assumptions) the operation of quantum apparatus, including entire circuits.   Recently, Aharonov et al.\ \cite{Aharonov:2008:Interactive-Pro} and Broadbent et al.\ \cite{Broadbent:2008:Universal-blind} considered verifiers with limited quantum capabilities.  As well, Broadbent et al.\ demonstrated an interactive protocol, which has the full power of $\mathsf{BQP}$, between two entangled, but non-communicating, efficient quantum provers and an efficient classical verifier.  Ideally we would like to show that that every problem in $\mathsf{BQP}$ has an interactive proof with a single efficient quantum prover and a classical verifier (let us call the class of such problems $\mathsf{IP_{BQP}}$.)  Unfortunately, self-testing and the Broadbent et al.\ protocol both rely fundamentally upon non-local properties of quantum theory, making the techniques unsuitable to the single-prover scenario.

In this context, the current work is interesting because it shows that there is an oracle relative to which $ \mathsf{IP_{BQP}} \nsubseteq \mathsf{MA}$\footnote{Note that the more general result of Hallgren and Harrow \cite{Hallgren:2008:Superpolynomial} does not have this property, since the oracle in their construction can be used to verify a certificate.}.  Since $\mathsf{BQP} \subseteq \mathsf{PSPACE} = \mathsf{IP}$ an interactive proof exists for every problem in $\mathsf{BQP}$, but the only known construction, due to Shamir \cite{Shamir:1992:IP--PSPACE}, uses a $\mathsf{PSPACE}$ prover.  The interest for the current result, then, is the fact that the interactive proof has an efficient quantum prover.

Another interesting aspect of this work is the fact that the structure of the interactive proof is not informed by the structure of the prover but arises naturally from the structure of the problem.  This is in contrast to results in \cite{Magniez:2006:Self-testing-of} and \cite{Broadbent:2008:Universal-blind}, which essentially analyze the quantum provers in action to verify correct operation.  This may indicate that a different methodology from these partial results is necessary, or at least useful, in order to reduce the number of provers down to one.

\bibliographystyle{halphads}
\bibliography{Global_Bibliography}
\end{document}